\documentclass{article}

\usepackage{amssymb}
\usepackage{amsmath}
\usepackage{amsthm}
\usepackage[ruled,vlined]{algorithm2e}
\usepackage[dvipsnames]{xcolor}

\usepackage{tikz}
\usetikzlibrary{shapes,snakes}
\usepackage{todonotes}
\usepackage{graphicx}

\newtheorem{proposition}{Proposition}
\newtheorem{lemma}{Lemma}
\newtheorem{theorem}{Theorem}
\newtheorem{definition}{Definition}

\newcommand{\OPT}{\mathrm{OPT}}
\newcommand{\ALG}{\mathrm{ALG}}

\newcommand{\prob}[1]{\mathbb{P}\left[ #1 \right]}

\newcommand{\Ex}[1]{\mbox{\rm\bf E}\left[#1\right]}

\title{A Collection of Lower Bounds for Online Matching on the Line}

\author{Antonios Antoniadis\thanks{Universit\"at des Saarlandes and Max Plank Institut f\"ur
Informatik, Saarland Campus, Germany. \texttt{aantonia@mpi-inf.mpg.de}. Supported by the Deutsche Forschungsgemeinschaft (DFG, German Research Foundation) under AN 1262/1-1.}
\and Carsten Fischer \thanks{Department of Computer Science, University of Bonn, Germany. \texttt{carsten.fischer@uni-bonn.de}. Supported by ERC Starting Grant 306465 (BeyondWorstCase).}
\and Andreas T\"onnis \thanks{Departamento de Ingenier\'ia Matem\'atica, Universidad de Chile, Chile. \texttt{atoennis@dim.uchile.cl}. Supported by Conicyt PCI PII 20150140. Work was done while the author was employed at the University of Bonn. Supported by ERC Starting Grant 306465 (BeyondWorstCase).}}

\begin{document}
\maketitle

\begin{abstract}

In the online matching on the line problem, the task is to match a set
of requests $R$ online to a given set of servers $S$. The distance
metric between any two points in $R\,\cup\, S$ is a line metric and the
objective for the online algorithm is to minimize the sum of distances
between matched server-request pairs. This problem is well-studied and
-- despite recent improvements --
there is still a large gap between the best known lower and upper bounds:
The best known deterministic algorithm for the problem is $O(\log^2n)$-competitive,
while the best known deterministic lower bound is $9.001$. The lower
and upper bounds for randomized algorithms are $4.5$ and $O(\log n)$ respectively.

We prove that any deterministic online algorithm which in each round:
$(i)$ bases the matching decision only on information \emph{local} to
the current request, and $(ii)$ is \emph{symmetric} (in the sense that
the decision corresponding to the mirror image of some instance $I$ is
the mirror image of the decision corresponding to instance $I$), must
be $\Omega(\log n)$-competitive. We then extend the result by showing
that it also holds when relaxing the symmetry property so that the
algorithm might prefer one side over the other, but only up to some
degree.  This proves a barrier of $\Omega(\log n)$ on the competitive
ratio for a large class of ``natural'' algorithms. This class includes
all deterministic online algorithms found in the literature so far.

Furthermore, we show that our result can be extended to randomized
algorithms that locally induce a symmetric distribution over the
chosen servers. The $\Omega(\log n)$-barrier on the competitive ratio
holds for this class of algorithms as well. 
\end{abstract}


\section{Introduction}

The \emph{online matching on the line problem (OML)} is a notorious
special case of the \emph{online metric matching problem (OMM)}. In
both problems a set of servers $\{s_1, s_2, \dots, s_n\} =: S$ is
initially given to the algorithm, then requests from a set $R := \{r_1, r_2,
\dots, r_n\}$ arrive online one-by-one. In OMM all servers and requests
are points in an arbitrary metric, while in OML all points,
$r_i\in R$ and $s_j\in S$, correspond to numbers in $\mathbb{R}$ and
the distance between two such points is given by the line metric, that
is the 1-dimensional Euclidian metric, i.e., $d(r_i, s_j) = |r_i -
s_j|$.  Whenever a request $r$ arrives it has to be
matched immediately and irrevocably to an unmatched server $s$, and
the edge $e=\langle r,s\rangle$ gets added to the matching $M$. The
objective for the online algorithm is to minimize the sum of distances
between matched server-request pairs $\min\sum_{e\in M} d(e)$.

The original motivation for OML is a scenario where items of
different size have to be matched in an online fashion. For example,
in a ski rental shop, skis have to be matched online to customers of
approximately the same size. In this setting, customers arrive online
one-by-one and have to be served immediately. A set
of skis/customers of size $x$ can be represented by a server/request at
point $x$ on the real line, and the length of an edge in the matching
would represent the ``amount of mismatch'' between the corresponding pair of
skis and the customer.

Since the concepts of local and local symmetric algorithms will be important
for the discussion below as well as our results, we start by introducing the necessary definitions. 
 When matching a request $r$ to a server, one can restrict the server
 choice to the nearest free servers $s_L$ and $s_R$, placed on the left
 and right of $r$ respectively.  We call these two servers the
 \emph{surrounding servers} for request $r$. It is easy to see, by an
 exchange argument, that any online algorithm can be converted to one
 that chooses among the surrounding servers for each request -- without
 increasing the competitive ratio.

\begin{figure}[h]
\centering
\includegraphics[page=2,width=1.0\textwidth]{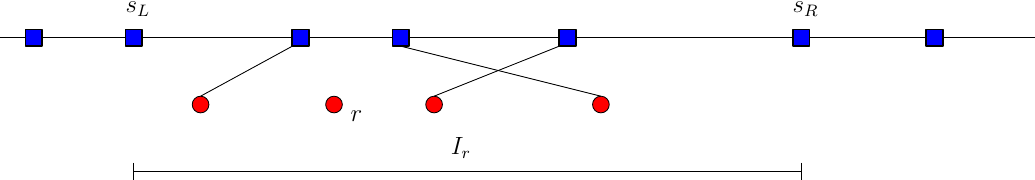}
\caption{We use red circles to represent requests and blue squares to represent servers. Lines show to which server each request is matched to. For a clearer representation we will often arrange servers and requests on different lines. It is without loss of generality to consider algorithms that match a request $r$ to one of its surrounding servers $s_L$ or $s_R$.}
\label{fig:local}
\end{figure}

Koutsoupias and Nanavati introduced the concept of
\emph{local} algorithms which will play a central role in our results:

\begin{definition}[\cite{DBLP:conf/waoa/KoutsoupiasN03}]
  Let $s_L$ and $s_R$ be the surrounding free servers for request $r$. An
  online algorithm is called local if it serves $r$ with one of $s_L$
  and $s_R$ and furthermore the choice is based only upon the history
  of servers and requests in the \emph{local-interval} $I_r = [s_L,
  s_R]$ of $r$.
\end{definition}

Note that this also implies that local algorithms are invariant with
respect to parallel translation of server and request locations. Since
the algorithm is only allowed to use local information, it can only
use relative positions within an interval and not absolute positions of
servers or requests.

One can restrict the class of local algorithms by introducing the
concept of \emph{local symmetric algorithms}. The main idea is that a
local algorithm is also symmetric, if mirroring the whole interval
$[s_L,s_R]$, also causes the algorithm to ``mirror'' its server
choice. 

\begin{definition}
  Let $A$ be a local algorithm, and consider the arrival of a request
  $r$ with a local interval $I_r=[s_L,s_R]$. 
Let
  $m = (s_L + s_R)/2$ be the point in the middle of interval $I_r$,
  and let $I_r'$ be the reflection of $I_r$ across point $m$.
  We say that $A$ is \emph{local symmetric}, if for any interval $I_r$
  when it chooses to match $r$ to a server $s\in \{s_L,s_R\}$ then
  it chooses to match the mirror image of $r$ to the mirror image of
  $s$ in $I_r'$.
\label{def:local_symmetric}
\end{definition}

See Figure~\ref{fig:symmetric} for an example of a symmetric
algorithm.

\begin{figure}[h]
\centering
\includegraphics[width=1.0\textwidth]{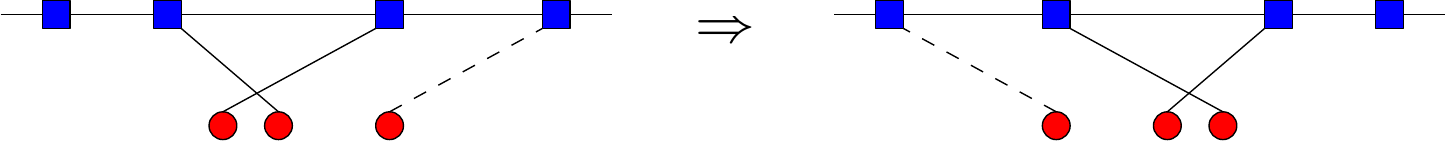}
\caption{The dashed line represents the choice of the algorithm. If we reflect the local-interval $[s_L,s_R]$, a local symmetric
    algorithm will ``reflect'' its decision.}
\label{fig:symmetric}
\end{figure}

A careful reader might have noticed that
Definition~\ref{def:local_symmetric} does not specify how a local
symmetric algorithm behaves in the case where intervals $I_r$ and
$I_r'$ happen to be indistinguishable\footnote{This could be resolved by for
example letting the algorithm choose the server arbitrarily or 
allow the adversary to force the server selection in this specific
border case.}. However this is not important in the scope of this paper,
since all of our constructions only use intervals $I_r$ and $I_r'$
that are clearly distinguishable from each other.

As we will see in the next subsection, a generalization of the class of local
symmetric algorithms contains all studied algorithms in the literature
for the problem.

\subsection{Related Work}
\subsubsection{The story so far.}
For the more general OMM problem, where the underlying metric is not
restricted to the line, the best possible competitive ratio is $2n-1$,
and algorithms that attain this ratio were analyzed
in~\cite{DBLP:journals/tcs/KhullerMV94,DBLP:journals/jal/KalyanasundaramP93,DBLP:conf/approx/Raghvendra16}.
In essence all three results employ a variant of the same online
\emph{$t$-net-cost algorithm}. The $t$-net cost algorithm, for each
round $i$, calculates a specific offline matching $M_{i-1}$ among the
server set and the $i-1$ first requests and then finds the minimum
$t$-net cost of an augmenting path on $M_{i-1}$ and the current
request. The $t$-net-cost augmenting path is given by weighting the
forward edges in the augmenting path by parameter $t$ before
subtracting the backward edges. This identifies a new free server to
which the current request is matched. Khuller, Mitchell and Vazirani,
as well as Kalyanasundaram and
Pruhs~\cite{DBLP:journals/tcs/KhullerMV94,DBLP:journals/jal/KalyanasundaramP93}
independently studied the variant of this algorithm for $t=1$ and with
$M_{i-1}$ being the optimal offline matching on the first $i-1$
requests. They therefore augment with the classical Hungarian method,
while Raghvendra~\cite{DBLP:conf/approx/Raghvendra16} employs
$t=n^2+1$ and a more complex offline matching $M_{i-1}$.  One can
easily show that all variants of the $t$-net cost algorithm are local
symmetric when applied on the line metric.

However, in the more special case of line metrics, the linear lower bound does
not hold. For a long time, the best known lower bound for the
problem was $9$ and it was conjectured to be tight.  But in 2003, it
was shown, that no deterministic algorithm for OML could be better
than $9.001$-competitive \cite{DBLP:journals/tcs/FuchsHK05} and this
remains the best known lower bound to date.

Regarding upper bounds in the line metric, Koutsoupias and
Nanavati~\cite{DBLP:conf/waoa/KoutsoupiasN03} studied the \emph{work
  function algorithm (WFA)}, and showed that this is also
$O(n)$-competitive, along with a lower bound of $\Omega(\log n)$ on
its competitive ratio. Intuitively, WFA tries to balance out the
$1$-net-cost algorithm with the greedy algorithm that matches each
request to the closest available server upon arrival. Koutsoupias and
Nanavati, conjectured that WFA is $\Theta(\log n)$-competitive, but
whether that is the case or not remains an open question. Although it
is not straightforward that the work function algorithm is local, this
was proven in~\cite{DBLP:conf/waoa/KoutsoupiasN03}, and symmetry
easily follows by the definition of the algorithm, placing WFA in the
class of local symmetric algorithms.

The first deterministic algorithm to break the linear competitive
ratio was the \emph{$k$-lost cows algorithm ($k$-LCA)} by Antoniadis
et al.~\cite{DBLP:conf/waoa/AntoniadisBNPS14}. Their algorithm uses a
connection of the matching problem to a generalization of the
classical lost cow search problem from one to a larger number of
cows. They give a tight analysis and prove that their algorithm is
$\Theta(n^{0.58})$-competitive.
It is easy to verify that the $k$-lost cow algorithm is local. However
by design, $k$-LCA has a very slight ``bias'' towards one direction
and is therefore not symmetric. A slight generalization of our
construction for local symmetric algorithms is enough to capture this
algorithm as well.

Very recently, Nayyar and Raghvendra~\cite{focs_paper} studied the
$t$-net cost algorithm for a constant $t>1$ on the line metric
(actually the considered setting is slightly more general). Through a
technically involved analysis they showed that for such values of $t$
the $t$-net-cost algorithm is $O(\log^2 n)$-competitive. As already
discussed, this algorithm is also local symmetric.

Finally, it is worth noting, and can be easily verified, that the
\emph{greedy algorithm} which matches each request to the closest free
server is $\Omega(2^n)$-competitive.

Regarding randomized algorithms for the more general online metric
matching problem, there are several known sub-linear
algorithms. Meyerson et al.~gave a randomized greedy algorithm that uses a tree embedding of the metric space and is $O(\log^3(n))$-competitive~\cite{DBLP:conf/soda/MeyersonNP06}.
Bansal et
al.~refined this approach and gave a $O(\log^2(n))$-competitive
algorithm~\cite{DBLP:journals/algorithmica/BansalBGN14}. More
recently, Gupta and Lewi gave two greedy algorithms based on tree
embeddings that are $O(\log(n))$-competitive for doubling metrics -- and therefore
also the line metric. Additionally, they analyzed the randomized
harmonic algorithm and showed that it is $O(\log(n))$-competitive for
line metrics~\cite{DBLP:conf/icalp/GuptaL12}. All these algorithms
locally induce a symmetric distribution over the chosen
servers. 

An extensive survey of the OML problem can be found under~\cite{DBLP:journals/sigact/Stee16} and~\cite{DBLP:journals/sigact/Stee16a}.

\subsubsection{Other Related Problems.}
A closely related problem is the transportation problem, which is a
variant of online metric matching with resource
augmentation. Kalyanasundaram and Pruhs~\cite{KalyanasundaramP00}
showed that there is a $O(1)$-competitive algorithm for this problem
if the algorithm can use every server twice, whereas the offline
benchmark solution only uses each server once. Subsequently, Chung et
al.~\cite{DBLP:conf/latin/ChungPU08} gave a polylogarithmic algorithm
for the variant where the algorithm has one additional server at every
point in the metric where there is at least one server.

Another problem which resembles online matching is the $k$-server
problem (see~\cite{Koutsoupias09} for a survey). The main difference
between the two problems is that in the $k$-server problem a server
can be used to serve subsequent requests, while in the online matching
problem a server has to be irrevocably matched to a request.

\subsection{Our Contribution}

Our first result is showing that any deterministic algorithm, that
is local symmetric, has to be $\Omega(\log n)$-competitive.
 
As already mentioned the class of local symmetric algorithms includes
all known algorithms except for the $k$-LC algorithm. However we are
able to generalize the construction of our lower bound instance so
that it contains an even wider class of algorithms -- including
$k$-LC.  The main implication of our work is that new algorithmic
insights are necessary if one hopes to obtain a $o(\log
n)$-competitive algorithm for the problem.

Our lower bound instance can be seen as a full binary tree with carefully chosen
distances on the edges. We describe the construction of this tree
recursively. The instance is designed in such a way, that every
request arrives between two subtrees, and the algorithm always has to
match it to one of the two furthest leafs of these subtrees. This
incurs a cost of $\Omega(n)$ at each of the $\log(n)$ levels of the
tree. In
contrast, the optimal solution always matches a request to
neighboring server for a total cost of $O(n)$.

We complement these results with propositions that showcase the power
and limitations of our construction. First, we show that any algorithm
that, for every level of the lower bound instance, has a bias bounded
by a factor of two with respect to the largest bias on the previous
levels fulfills the conditions of our main theorem. To beat our lower
bound instance an algorithm would require an asymmetric bias in its
local decision routine and furthermore this bias would have to grow
exponentially by more than a factor of two as the depth of the
instance increases. To the best of our knowledge, the only known
algorithm that features a local bias is the \emph{$k$-LC} algorithm
by Antoniadis et al.~\cite{DBLP:conf/waoa/AntoniadisBNPS14} and its
bias is only $(1+\epsilon)$. 

We denote that it seems hard to conceptualize a ``reasonable'' local
algorithm for the problem that has a bias greater than two and we
therefore believe/conjecture that $\Omega(\log(n))$ is likely a lower
bound for the even broader class of local algorithms.

Furthermore, we show that the lower bound also applies to a wide class
of randomized algorithms. If the instance can be tailored to the
algorithm in such a way, that the algorithm induces a symmetric
distribution over the free servers on each local subinstance, then an
analogous lower bound of $\Omega(\log n)$ holds true. For this result,
we have the same conditions on the bias of the algorithm as in the
deterministic setting. We show that these conditions are fulfilled by
the \textsc{Harmonic}-algorithm introduced by Gupta and Lewi
\cite{DBLP:conf/icalp/GuptaL12}. This simple randomized algorithm is
known to be $\Theta(\log n)$-competitive. 

\section{Lower Bounds for Deterministic Algorithms}

This section is devoted to our main results for deterministic
algorithms. First, we show that any local symmetric algorithm must
be $\Omega(\log n)$-competitive. We already saw that this class of
algorithms is broad and captures all known deterministic algorithms
except the $k$-LCA. Then we extend the construction of our lower bound
towards local algorithms that have a limited asymmetric bias in their
decision routine.

The construction for both proofs resembles a full binary tree. It is
defined recursively such that the behavior of the online algorithm on
any subtree exactly mirrors its behavior on the sibling subtree. In
order to achieve this, we will define for each level of the tree,
intervals which consist of a new request located between two subtrees
of one level lower so that the only two free servers in the interval
are the ones furthest from the current request. We start at level one
with simple intervals that consist of two servers and one request
roughly in the center between the servers. 
The algorithm, by locality, will have to match the current request to one of the free servers. If, in one subtree, the algorithm matches to the right, we can create a similar subtree where the algorithm matches to the left (and vice versa) by only marginally changing the distances.
We recursively repeat this construction while
ensuring that, for any two sibling subtrees, the one on the left has
the leftmost server free and the one on the right the rightmost one.

We start with the special case of local symmetric algorithms and give
the formal construction of the lower bound instance. 

At the base level, we have trees $T_0$ and $\overline{T}_0$ that
contain a single server each. On level $i\in\{1, \dots, k\}$, we
combine two trees: $T_{i-1}$ which has its leftmost
server free, and
$\overline{T}_{i-1}$ which has its rightmost server free, into an interval. We create an
interval with $T_{i-1}$ followed by request $r_i$ at a distance $1$,
which is in turn followed by $\overline{T}_{i-1}$ at a distance
$1+\epsilon$ to the right of $r_i$. If the algorithm matches $r_i$ to
the free server on the right (resp. left) in this interval then this creates
tree $T_i$ (resp. $\overline{T}_i$), and by symmetry of the algorithm if we
swap the distances $1$ and $1+\epsilon$ around it will match $r_i$ to
the left (resp. right) thus creating tree $\overline{T}_{i}$ (resp. $T_i$).  We
will see that, up to the highest tree level, $T_i$ and $\overline{T}_i$ are
always a mirror image of each other.

It is helpful to define the interval of a tree $T$ as $I(T)$. $I(T)$
is the interval from the leftmost to the rightmost server in $T$
just before the request of $T$ was matched, i.e., $I(T)$ contains
exactly two free servers located at its endpoints, and exactly one
unmatched request contained between the two subtrees of the root of $T$.

\begin{theorem}
  Let $A$ be a local symmetric algorithm. Then, there exists an
  instance with $n$ servers such that
  the competitive ratio of $A$ is in $\Omega(\log(n))$.
\label{thm:main-thm}
\end{theorem}

\begin{proof}
  We will show by induction that the interval $I(T_k)$ is always a
  mirror image to interval $I(\overline{T_k})$ Therefore, local
  symmetric online algorithms will match in one of them (w.l.o.g.\ in
  $I(T_k)$) to the right and in the other interval to the left. In
  this way, for each request $r_i$ on the $i$-th level of the
  recursion, only the leftmost and rightmost servers in the
  respective interval are available.

  Since the trees $T_0$ and $\overline{T}_0$ are identical, it
  immediately follows
  that intervals $I(T_1)$ and $I(\overline{T}_1)$ are mirror
  images of each other. Again, due to symmetry of the online
  algorithm, we may assume that $T_1$ leaves the left server
  open and $\overline{T}_1$ leaves its right server open.

  Now, for the inductive step, intervals $I(T_i)$ and
  $I(\overline{T}_i)$ both take the same subtrees $T_{i-1}$ and
  $\overline{T}_{i-1}$ as building blocks and those subtrees are
  already mirror images of each other by the inductive
  hypothesis. Furthermore the distances between the subtrees and
  request $r_i$ are also a mirror image of one another (recall that in
  one tree these distances are $1$ and $1+\epsilon$ and in the other
  one $1+\epsilon$ and $1$). So, $T_i$ and $\overline{T}_i$ must also
  be mirror images of each other. In addition, since in $T_{i-1}$ the
  leftmost server is free and in $\overline{T}_{i-1}$ the rightmost
  server is free, we may adapt the construction so that $T_i$ also leaves the leftmost
  server free and $\overline{T}_i$ also leaves the rightmost server
  free.

  We have established that, when request $r_i$ arrives, the only free
  servers are the left and rightmost servers $s_L$ and $s_R$ of
  $I(T_i)$ (similarily also for $I(\overline{T}_i)$). By construction,
  the distances are $d(r_i, s_L) \geq 2^i - 1$ and $d(r_i, s_R) \geq
  2^i - 1$ because there are $2^{i-1}$ servers in the subtrees
  $T_{i-1}$ and $\overline{T}_{i-1}$, each at a distance of
  $2+\epsilon$ from each other, and all of them, except the outermost
  are already matched. In addition, there are $2^{k-i}$ requests on
  level $i$ in a tree of depth $k$. Meanwhile, the minimum distance
  between a request and a server is $1$, so the competitive ratio is
\[\frac{c(\ALG)}{c(\OPT)} = \frac{ \sum_{i=1}^k 2^{k-i}(2^i - 1)}{\sum_{i=1}^k 2^{k-i} 1} \geq \frac{k2^k - 1}{2^k-1} \geq k - \frac{1}{2^k}\in\Omega(k)\;.\]
\end{proof}

\subsection{Local \& Non-Symmetric Algorithms}\label{sec:deterministic_asymmetric}

We generalize the lower bound for local and symmetric algorithms that
was used to prove Theorem~\ref{thm:main-thm}. Towards this end,
we define a choice function $C:I \rightarrow \{s_L, s_R\}$ that takes
as input an interval $I$ along with an unmatched request
$r\in[s_L,s_R]$. The interval $I$ is such that the only free servers
it contains are $s_L$ and $s_R$ at the left and right end of the
interval respectively. In addition, $I$ includes information about all
other matched servers and requests in between $s_L$ and $s_R$. The
choice function returns $s_R$ if the algorithm decides to match $r$ to
the right and $s_L$ otherwise. By definition, every local algorithm
can be fully characterized by such a choice function.

Our construction follows a similar recursive structure starting with
trees $T_0$ and $\overline{T}_0$ that only contain a single
server. Then from level to level, we again combine two trees $T_{i-1}$
and $\overline{T}_{i-1}$ with a request $r_i$ in between to create a
tree $T_i$ or $\overline{T}_i$. The difference to the previous
construction is the distances between $r_i$ and the nearest server in
the neighboring subtrees. For both trees $T_i$ and $\overline{T}_i$,
let $a_i$ be the distance to the rightmost (and therefore closest)
server of the subtree $T_{i-1}$. Furthermore let $b_i$
(resp. $b_i+\epsilon$) be the distance to the left-most server of
subtree $T_{i-1}$ (resp. $\overline{T}_{i-1}$). Here, $a_i$ and $b_i$
are carefully chosen in such a way that $C(I(T_i)) = s_R$ and
$C(I(\overline{T}_i)) = s_L$. We may assume that such $a_i$ and $b_i$
do always exist, since otherwise we could set one of them to $1$ and
the other to $\infty$ resulting in an unbounded competitive ratio.

\begin{figure}
\centering
\includegraphics[page=4,width=1.0\textwidth]{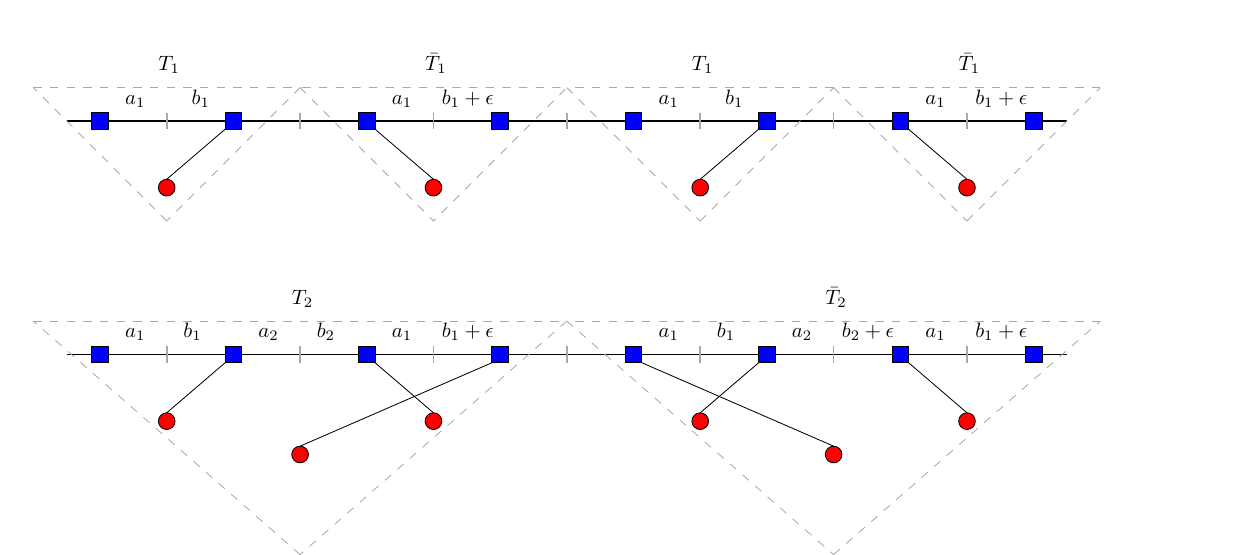}
\caption{The construction for the lower bound in Theorem~\ref{thm:deterministic_local}.}
\label{fig:lb_relsym}
\end{figure}

Our final interval consists of two trees $T_k$ and 
$\overline{T}_k$. To simplify the presentation, we skip the last
request in between them. In other words, we end the process while
still having two free servers. However this is without loss of
generality since the adversary can present two requests, each
collocated with one of the free servers, thus essentially ``removing''
these servers from the instance. A crucial difference to the previous
proof is that although we cannot leverage symmetry of the algorithm in
order to show that $r_i$ gets matched to opposite servers in $I(T_i)$
and $I(\overline{T}_i)$, we now get this property directly by our
choice of $a_i$'s and $b_i$'s. The analysis then follows that in the
proof of Theorem~\ref{thm:main-thm} but is significantly more involved since the
distances can now vary from level to level.

The main result of this section is the following theorem. After
proving it, we discuss its implications to specific classes
of algorithms.

\begin{theorem}
Fix a local online algorithm $A$ and let $x_i = a_i+b_i$, where $a_i$
and $b_i$ are defined for $A$ as described above. If there holds
\[k - \frac{\sum_{i=1}^{k} x_i 2^{-i} i}{\sum_{i=1}^{k} x_i 2^{-i}}\in\Omega(k)\;,\]
then algorithm $A$ is $\Omega(\log(n))$-competitive.
\label{thm:deterministic_local}
\end{theorem}

\begin{proof}
  It can be easily shown through an exchange argument (see
  also~\cite{DBLP:journals/sigact/Stee16}) that there is always an
  optimal solution that matches $(s_i, r_i)$ when the servers and
  requests are sorted by their position. 
Therefore, there is an optimal solution for the instance described above that matches every request on recursion level $i$ to the next server in the neighboring block $T_{i-1}$ or $\overline{T}_{i-1}$. So every request on level $i$ pays either $a_i$ or $b_i$. Thus the cost of the optimal  solution in the instance is the sum of the optimal solutions on $T_k$ and $\overline{T}_k$, which differ by at most an $\epsilon$.
\begin{align*}
c(\OPT) &= 2\cdot\min \left\{\sum_{i=1}^k 2^{k-i}a_i, \frac{\epsilon}{2} + \sum_{i=1}^k 2^{k-i}\left(b_i + \frac{\epsilon}{2}\right)\right\}\leq \sum_{i=1}^k 2^{k-i}(a_i + b_i + \epsilon)\;.
\end{align*}
With $\epsilon$ arbitrarily small, its contribution to the cost is negligible. For simplicity of notation, we omit all occurrence of $\epsilon$ in the rest of the proof. 

In contrast to the optimal solution, the online algorithm always matches the request $r_i$ to the right in $T_i$ and to the left in $\overline{T}_i$. By construction, the free server after request $r_i$ arrived in subtree $T_i$ is the left most server and respectively in $\overline{T}_i$ the right most server.

Thus on level $i$ the distance to the matched server is $d(r_i, s_R) = b_i + \sum_{j=1}^{i-1}2^{i-1-j}(a_j+b_j)$ and $d(s_L, r_i) = a_i + \sum_{j=1}^{i-1}2^{i-1-j}(a_j+b_j)$. The instance consists of two trees $T_k$ and $\overline{T}_k$, so the cost of the algorithms solution is
\begin{align*}
c(\ALG) &= \sum_{i=1}^k 2^{k-i}\left(d(s_L, r_i) + d(r_i, s_R)\right)=\sum_{i=1}^k 2^{k-i}\sum_{j=1}^i 2^{i-j}(a_j+b_j)\\
&=\sum_{i=1}^k\sum_{j=1}^i 2^{k-j}(a_j + b_j) =\sum_{i=1}^k (a_i+b_i)(2^{k+1-i} - 1)(k+1-i)\;.
\end{align*}

We relabel $a_i+b_i=x_i$, this gives us
\begin{align*}
\frac{c(\ALG)}{c(\OPT)}
&\geq \frac{\sum_{i=1}^k (a_i+b_i)(2^{k+1-i} - 1)(k+1-i)}{\sum_{i=1}^{k} x_i 2^{k-i}}\\
&\geq \frac{\sum_{i=1}^{k} (x_i 2^{-i}) (k - i)}{\sum_{i=1}^{k} x_i 2^{-i}}
= k - \frac{\sum_{i=1}^{k} x_i 2^{-i} i}{\sum_{i=1}^{k} x_i 2^{-i}}\;.
\end{align*}
\end{proof}

We give a sufficient condition for Theorem~\ref{thm:deterministic_local} that is easier to work with. If, for every level of the recursive construction, the bias of an online algorithm grows by at most a factor of two with respect to the maximal previous bias, then the online algorithm is $\Omega(\log n)$-competitive.

\begin{proposition}\label{prop:max}
A sufficient condition for Theorem~\ref{thm:deterministic_local} is $x_i\leq 2\max_{j\in[1:i-1]} x_j$.
\end{proposition}

Furthermore, we also show that this sufficient condition is nearly tight. If the choice function has a bias that is increasing by a factor of at least $2+\epsilon$ for some $\epsilon>0$ with each new recursive level of the instance, then we can only give a constant lower bound on the competitive ratio.

\begin{proposition}\label{prop:1+eps}
If, for an online algorithm $A$, the corresponding instance takes the form $x_i \geq (2+\epsilon)x_{i-1}$ with $x_0 = 0$ for $\epsilon > 0$ and for all $i\in[1:\log(n)]$, then the instance only proves a constant competitive ratio.
\end{proposition}

The proof of both proposition can be found in Appendix~\ref{app:non-symmetric}.


\section{Lower Bounds for Randomized Algorithms}

Our deterministic lower bound in Section~\ref{sec:deterministic_asymmetric} also extends to randomized local algorithms. The main difference is that, in the randomized case, we cannot deduce the right distances between requests and servers $a_i$ and $b_i$ from a deterministic choice function. Instead, we construct the instance in such a way that the position of the free server in every subtree $T_i$ is symmetrically distributed. Then it is easy to see that the algorithm is bound to lose at least a constant fraction in the competitive ratio over the algorithm in the deterministic case.

Again, the instance is constructed analogously to the previous section. The main difference is that now $T_i$ consists of two subtrees $T_{i-1}$ with an additional request $r$ in between. Similarly to before we set the distances between the subtrees and the new request as $d(T_{i-1}, r_i) = a_i$ and $d(r_i, T_{i-1})=b_i$ with the exception of the top level request $r_k$. On level $k$, let $a_k = b_k = 0$. By construction, and as before, every tree $T_i$ contains exactly one free server $s\in T_{i}$. For convenience of notation, $s\in[2^i]$ also denotes the position of $s$ within $T_i$.

\begin{theorem}
\label{thm:randomized} 
Consider any randomized online algorithm $A$, for which (i) $a_i$ and $b_i$ can be chosen in such a way that the distribution $p^i(s)$ over the position of the free server $s(T_i)\sim_{p^i}[2^i]$ is symmetric, and (ii) the distances $x_i = a_i + b_i$ fulfill \[k - \frac{\sum_{i=1}^{k-1}x_i2^{-i}i}{\sum_{i=1}^{k-1} x_i2^{-i}} \in \Omega(k)\;,\] for every $i\in[k-1]$. Algorithm $A$ is $\Omega(\log n)$-competitive on the instance $T_{\log n}$.
\end{theorem}

\begin{proof}
If the distribution $p^i$ over the position of the free server $s$ in $T_i$ is symmetrical, then with probability $\frac{1}{2}$ the server is in the first half of $T_i$. In this case $s\leq 2^{i-1}$. Analogously with probability $\frac{1}{2}$, we also have $s\geq 2^{i-1}+1$. If both events occur at the same time, then in $T_{i+1}$ the distance between matched server and request is at least $\min\{d(s_L, r_i), d(r_i, s_R)\} \geq \sum_{j=1}^{i-1} (a_j + b_j)2^{i-1-j}$. Here, we omit the cost of $a_i$ or $b_i$ because the algorithm will not pay both.

Therefore, the expected cost of the online algorithm is at least
\begin{align*}
\Ex{c(\ALG)} &\geq \sum_{i=1}^k \frac{1}{4} 2^{k-i}\sum_{j=1}^{i-1} (a_j + b_j)2^{i-1-j}\\
&= \frac{1}{4}\sum_{i=1}^{k-1} (a_i + b_i)(2^{k-i-1}-1)(k-i)\;.
\end{align*}

Similar to the previous section, the cost of the optimal solution are
\[c(\OPT) = \min\left\{\sum_{i=1}^k 2^{k-i}a_i, \sum_{i=1}^k 2^{k-i}b_i\right\} \leq \sum_{i=1}^k 2^{k-i-1}(a_i + b_i) \;.\]

Again we substitute $x_i = (a_i+b_i)$, then the lower bound instance guarantees a competitive ratio of at least
\begin{align*}
\frac{\Ex{c(\ALG)}}{c(\OPT)} &\geq \frac{\frac{1}{4}\sum_{i=1}^{k-1} x_i(2^{k-i-1}-1)(k-i)}{\sum_{i=1}^k 2^{k-i-1}x_i}
\geq \frac{\frac{1}{8}\sum_{i=1}^{k-1}x_i2^{-i}(k-i)}{\sum_{i=1}^{k-1} x_i2^{-i}}\;.
\end{align*}
In the last step, we use that $x_k = a_k + b_k = 0$. Now we have an expression similar to the previous proof, the same steps give the desired result.
\end{proof}

An example for a randomized online algorithm for OML is the
\textsc{Harmonic} algorithm by Gupta and
Lewi~\cite{DBLP:conf/icalp/GuptaL12}. They have shown that this
algorithm is $O(\log n)$-competitive in expectation. We show that
\textsc{Harmonic} fulfills the condition in
Theorem~\ref{thm:randomized}, and therefore provide an alternative
that \textsc{Harmonic} is $\Omega(\log n)$-competitive.

\begin{proposition}\label{prop:HarmonicLowerBound}
For the algorithm \textsc{Harmonic}, $a_i = b_i = 1$ yields a symmetric distribution $p^i(s)$ for all $i\in [k]$.
\end{proposition} 

The proof of this proposition can be found in Appendix~\ref{app:randomized}.


\section{Discussion}

This paper rules out an $o(\log(n))$-competitive ratio for a
wide class of both deterministic and randomized algorithms for
OML. This means that new algorithmic insights are necessary if one
hopes to obtain such a $o(\log(n))$-competive algorithm for the
problem. It is natural to try and conceptualize a ``reasonable''
deterministic/randomized algorithm that beats our instance.  As
already mentioned, we find it particularily hard to conceptualize such
a local algorithm and we therefore conjecture that the lower bound of
$\Omega(\log n)$ holds for all local algorithms, even though a
different construction would be required to handle algorithms with alternating and exponentially growing bias.
However, it
would be interesting to try to design and analyze a non-local
algorithm, that also employs information from outside of the local-interval
in order to match a request.

The best deterministic algorithm known so far is
$O(\log^2(n))$-competitive, and for many known algorithms the best
known lower bound on their competitive ratio is $\Omega(\log(n))$, it
would be reasonable to work on a tighter analysis for an
existing algorithm in order to (hopefully) prove it
$\Theta(\log(n))$-competitive. 
The WFA algorithm has been conjectured to be $\Theta(\log n)$-competitive before and our work does not change anything on that front. Another promising candidate is the $t$-net-cost algorithm for some $t>1$ because this is the currently best known algorithm and it is not obvious that the analysis is tight for the line metric.

\bibliographystyle{plain}
\bibliography{references}

\begin{appendix}

\section{Missing proofs in Section~\ref{sec:deterministic_asymmetric}}\label{app:non-symmetric}

\subsection{Proof of Proposition~\ref{prop:max}}

We want to show that the condition $x_i \leq 2 \max_{j \in [1:i-1]} x_j$ is sufficient for a logarithmic lower bound.

\begin{lemma}
Let $(x_i)_{i \in [1:k]}$ be a sequence that safisfies the conditions
\begin{enumerate}
\item $x_1 > 0$;
\item $x_i \geq 0$ for $i \in [1:k]$;
\item and $x_i \leq 2 \max_{j \in [1:i-1]} x_j$.
\end{enumerate}
Then, we have for each $m \in [1:k]$
\begin{align*}
\frac{\sum_{i=1}^m 2^{-i}x_i}{\sum_{i=1}^k 2^{-i}x_i} \geq \frac{m}{k}\;.
\end{align*}
\end{lemma}

\begin{proof}
We will show for $m \in [1:k-1]$ that we have
\begin{align*}
\frac{1}{m} \sum_{i=1}^m x_i 2^{-i}
\geq \frac{1}{m+1} \sum_{i=1}^{m+1} x_i 2^{-i}\;.
\end{align*}
The statement then follows by repeated application of this identity.

Using basic calculations we can rewrite this as follows
\begin{align*}
\sum_{i=1}^m 2^{m-i} x_i \geq \frac{m}{2} \cdot x_{m+1}\;.
\end{align*}
Since we allow $x_{m+1}$ to be as large as $2\max_{1 \leq i \leq m} x_i$, it satisfies to show
\begin{align*}
\sum_{i=1}^m 2^{m-i} x_i \geq m \cdot \max_{1 \leq i \leq m} x_i\;.
\end{align*}
Let $i_1 < i_2 < \ldots < i_\ell$ denote the longest subsequence such that $x_{i_1} < x_{i_2} < \ldots < x_{i_\ell}$. Note, that $x_{i_\ell} = \max_{i \in [m]} x_i$. Furthermore, we set $i_{\ell+1} := m+1$.

We make the following observation: Let $j < \ell-1$. Then we have
\begin{align}
\label{eqn:recursiveEstimation}
\frac{2^{m-i_j}x_{i_j}}{i_{j+1}-i_j} \geq \frac{2^{m-i_{j+1}}x_{i_{j+1}}}{i_{j+2}-i_{j+1}}\;.
\end{align}
We can see this as follows: Rewriting the expression and using that $x_{i_{j+1}} \leq 2x_{i_j}$ we obtain
\begin{align*}
2^{i_{j+1}-i_j} \geq \frac{i_{j+1}-i_j}{i_{j+2}-i_{j+1}}\;.
\end{align*}
We see that the right hand side is maximized if the denominator is equal to $1$. Therefore, we obtain the estimate $2^{i_{j+1}-i_j} \geq i_{j+1}-i_j$. But this is clear, since $i_{j+1}-i_j$ is a natural number.

Then, a repeated application of (\ref{eqn:recursiveEstimation}) yields
\begin{align*}
\sum_{i=1}^m 2^{m-i}x_i
& \geq \sum_{u=1}^{\ell} 2^{m - i_u} x_{i_u} \\
& = 2^{m-i_1} x_{i_1} + \sum_{u=2}^\ell 2^{m - i_u} x_{i_u} \\
& \geq \frac{i_2 - i_1}{i_3 - i_2} 2^{m-i_2} x_{i_2} + \sum_{u=2}^\ell 2^{m - i_u} x_{i_u} \\
& = \left(\frac{i_2 - i_1}{i_3 - i_2} + 1 \right) 2^{m-i_2} x_{i_2} + \sum_{u=3}^\ell 2^{m - i_u} x_{i_u} \\
& = \frac{i_3 - i_1}{i_3 - i_2} 2^{m-i_2} x_{i_2} + \sum_{u=3}^\ell 2^{m - i_u} x_{i_u} \\
& \geq \ldots \\
& \geq \frac{i_{\ell+1}-i_1}{i_{\ell+1}-i_\ell} 2^{m-i_\ell} x_{i_\ell}\;.
\end{align*}

At first consider the case that $i_{\ell} = m$. Then, the sum is lower bounded by $\frac{(m+1)-1}{(m+1)-m} 2^{m-m} x_m = m x_m$. Now assume that $i_\ell < m$. Then we want to show that
\begin{align*}
\frac{m+1-1}{m+1-i_\ell} 2^{m - i_\ell} \max_{i \in [m] } x_i
\geq m \max_{i \in [m] } x_i\;.
\end{align*}
But this is true if $2^{m-i_\ell} \geq 1 + m- i_\ell$. Since $m > i_\ell$ due to our assumption, this holds true. Therefore, the statement follows.
\end{proof}

\begin{proof}[of Proposition~\ref{prop:max}]
Now we can upper bound the expression
\begin{align*}
\frac{\sum_{i=1}^k (2^{-i}x_i)i}{\sum_{i=1}^k 2^{-i}x_i}\;.
\end{align*}

It follows from the previous lemma that at least half of the mass of the probability distribution is located on the set $\{1, \ldots, \lceil k/2 \rceil\}$. Therefore, we have
\begin{align*}
\frac{\sum_{i=1}^k (2^{-i}x_i)i}{\sum_{i=1}^k 2^{-i}x_i}
\leq \lceil k/2 \rceil / 2 + k/2
\leq 3k/4 + 1/2\;.
\end{align*}
\end{proof}

\subsection{Proof of Proposition~\ref{prop:1+eps}}

For the proof of this proposition, we require the following technical lemma.

\begin{lemma}\label{lemma:sum}
Let $y_i \geq c\cdot y_{i-1}$ with $c>1$ for all $i\in [2:k]$, then we have
\[\frac{\sum_{i=1}^k y_i i}{\sum_{i=1}^k y_i} \geq \frac{\sum_{i=1}^k c^i i}{\sum_{i=1}^k c^i}\;.\] 
\end{lemma}

\begin{proof}
We multiply both sides with the denominators and reorder the sums such that we can apply the condition for all $i$.

\begin{alignat*}{3}
&\sum_{i=1}^k y_i i \sum_{j=1}^k c^j &&\geq \sum_{i=1}^k c^i i \sum_{j=1}^k y_j\\
\Leftrightarrow\qquad & \sum_{i=1}^k \sum_{j=1}^k y_i(i-j)c^j && \geq 0\\
\Leftrightarrow\qquad & \sum_{i=1}^k\sum_{j=1}^{i-1} y_i j c^{i-j} && \geq \sum_{i=1}^k\sum_{j=1}^{k-i} y_i j c^{i+j}\;.
\end{alignat*}
Up to this point, we rework the sum such that all $y_ic^j$ pairs only on one side of the inequality. 

The condition $y_i \geq cy_{i-1}$ also implies $y_i \geq c^jy_{i-j}$. We apply this inequality to all terms on the left-hand side $j$ times and change the order of summation twice to get
\begin{align*}
\sum_{i=1}^k\sum_{j=1}^{i-1} y_i j c^{i-j} &\geq \sum_{i=1}^k\sum_{j=1}^{i-1} y_{i-j} j c^{i}\\
&=\sum_{i=1}^k\sum_{j=1}^{i-1}y_{k-i+1}jc^{k-i+j+1}\\
&=\sum_{i=1}^k\sum_{j=1}^{k-i} y_i j c^{i+j}\;.
\end{align*}
\end{proof}

\begin{proof}[of Proposition~\ref{prop:1+eps}]
From the proof of Theorem~\ref{thm:deterministic_local}, we know that
\begin{align*}
\frac{c(\ALG)}{c(\OPT)} &\leq \frac{\sum_{i=1}^k x_i(2^{k+1-i} - 1)(k+1-i)}{\sum_{i=1}^k x_i2^{k-i}}\\
&\leq 2(k+1) - 2\cdot\frac{\sum_{i=1}^k x_i 2^{-i} i}{\sum_{i=1}^k x_i 2^{-i}}\;.
\end{align*}

We substitute $\frac{x_i}{2^i} = y_i$, then Lemma~\ref{lemma:sum} gives us
\[\frac{\sum_{i=1}^k x_i2^{-i}i}{\sum_{i=1}^{k}x_i2^{-i}} =\frac{\sum_{i=1}^k y_i i}{\sum_{i=1}^k y_i} \geq \frac{\sum_{i=1}^k (1+\epsilon)^i i}{\sum_{i=1}^k (1+\epsilon)^i}\;,\]
with $y_i \geq (1+\epsilon)y_{i-1}$ for all $i\in[2:k]$.

\begin{align*}
\frac{\sum_{i=1}^k (1+\epsilon)^i i}{\sum_{i=1}^k (1+\epsilon)^i}
&= \frac{k(1+\epsilon)^{k+1}-(k+1)(1+\epsilon)^k + 1}{\epsilon((1+\epsilon)^k - 1)}\\
&= \frac{(k\epsilon - 1)(1+\epsilon)^{k+1} + 1 + \epsilon}{\epsilon((1+\epsilon)^{k+1}-1)}\\
&= -\frac{1}{\epsilon} + \frac{k(1+\epsilon)^{k+1} + 1}{(1+\epsilon)^{k+1} - 1}\;.
\end{align*}
This tends to $k - \frac{1}{\epsilon}$ as $k$ tends to infinity. Therefore we have
\[\frac{c(\ALG)}{c(\OPT)} \leq 2(k+1) - 2\left(k-\frac{1}{\epsilon}\right) = 2 + \frac{1}{\epsilon}\in O(1)\;.\]
\end{proof}


\section{Proof of Proposition~\ref{prop:HarmonicLowerBound}}\label{app:randomized}

Without loss of generality, the left-most server of the instance is at position $2$. So, with $a_i = b_i = 1$, our set of servers is given by $\{2, 4, \ldots, 2^{k+1}\}$. Furthermore, there are only $2^{k}-1$ requests, so one server will be remained unmatched. Obviously the server that is left open in the optimal solution could also be matched for no additional cost, but for our lower bound we ignore this cost for the online algorithm.

\begin{proof}
In $T_k$ the last request $r$ will arrive at position $2^k+1$. For $s \in \{2, \ldots, 2^{k+1}\}$ we denote by $\tilde{s}$ the position of the server if we mirror $s$ at $r$, that is $\tilde{s} = 2r - s$.

We will show via induction that $p^k(s)$ is symmetric, that is, we have $p^k(s) = p^k(\tilde{s})$.

We start with the base case $k=1$. In this case, both servers have distance $1$ to the request. Therefore, the algorithm chooses both servers with equal probability. So we have $p^1(2) = p^1(4) = 1/2$.

We proceed with the inductive step: So our instance $T_{k+1}$ is given by two subinstances from the previous step on the sets $L := \{2, \ldots, 2^{k+1}\}$ and $R := \{2^{k+1} + 2, \ldots, 2^{k + 2}\}$ and the new request arrives at position $2^{k+1} + 1$. From the induction hypothesis we know that we have a symmetric distribution $p_L^k$ on $L$ and a symmetric distribution $p_R^k$ on $R$. 

We observe that the following two points are true:
\begin{itemize}
\item For $s \in L$ it is $p_L^k(s) = p_R^k(s + 2^{k+1})$, and for $s \in R$ it is $p_R^k(s) = p_L^k(s - 2^{k+1})$. This follows from the construction of our instance.
\item For $s \in R$ we have $p_L^k(s - 2^{k+1}) = p_L^k(\tilde{s})$, and for $s \in L$ we have $p_R^k(s + 2^{k+1}) = p_R^k(\tilde{s})$. This follows from the induction hypothesis.
\end{itemize}

From the definition of \textsc{Harmonic} it follows that we have for $x \in L$ and $y \in R$
\begin{align*}
& \prob{\left. \text{$r$ is matched to $y$} \right| \text{$x,y$ are unmatched}} \\
= & \prob{\left. \text{$r$ is matched to $\tilde{y}$} \right| \text{$\tilde{x},\tilde{y}$ are unmatched}}.
\end{align*}

Therefore, we have
\begin{align*}
p^{k+1}(x)
& = p_L^k(x) \sum_{y \in R} p_R^k(y) \prob{\left. \text{$r$ is matched to $y$} \right| \text{$x,y$ are unmatched}} \\
& = p_R^k(\tilde{x}) \sum_{y \in R} p_L^k(\tilde{y}) \prob{\left. \text{$r$ is matched to $\tilde{y}$} \right| \text{$\tilde{y}, \tilde{x}$ are unmatched}} \\
& = p_R^k(\tilde{x}) \sum_{y \in L} p_L^k(y) \prob{\left. \text{$r$ is matched to $y$} \right| \text{$y, \tilde{x}$ are unmatched}} \\
& = p^{k+1}(\tilde{x}).
\end{align*}
\end{proof}
\end{appendix}
\end{document}